\documentclass[10pt]{article}
\usepackage{mathrsfs}
\usepackage{amsfonts,enumerate}
\usepackage{amsmath,amssymb,amsthm,graphicx,verbatim,psfrag,url,epstopdf}
\def\S{\mathcal {S}}
\def\M{\mathcal {M}}
\def\N{\mathcal {N}}
\def\I{\mathcal {I}}
\def\A{\mathcal {A}}

\def\H{\mathcal {H}}
\def\X{\mathcal {X}}

\def\P{\mathscr{P}}
\def\C{\mathscr{C}}
\def\K{\mathscr{K}}
\def \p {\mathbf{p}}

\newcommand{\Rmnum}[1]{\expandafter\@slowromancap\romannumeral #1@}

\newtheorem{thm}{Theorem}
\newtheorem{clm}{Claim}
\newtheorem{lem}{Lemma}
\newtheorem{dfn}{Definition}
\newtheorem{obs}{Proposition}

\newtheorem{conjecture}{Conjecture}

\begin{document}
\title{Maxwell-independence: a new rank estimate for the $3$-dimensional generic rigidity matroid}
\author{Jialong Cheng\footnotemark[1]{  }\footnotemark[2]{}, Meera Sitharam\footnotemark[1]{}}
\date{}

\footnotetext[1]{University of Florida, Gainesville, FL 32611-6120, US, research supported in part by NSF grant DMSO714912 and
CCF 0610136, and a gift from SolidWorks Corporation
}

\footnotetext[2]{corresponding author: jicheng@cise.ufl.edu}
\maketitle

\begin{abstract}
The problem of combinatorially determining the rank of the 3-dimensional
bar-joint {\em rigidity matroid} of a graph is an important open problem in
combinatorial rigidity theory. Maxwell's condition states that the edges of a graph $G=(V, E)$ are {\em independent} in its $d$-dimensional generic rigidity matroid only if $(a)$ the number of edges $|E|$ $\le$ $d|V| - {d+1\choose 2}$, and $(b)$ this holds for every induced subgraph with at least $d$ vertices. We call such graphs {\em Maxwell-independent} in $d$ dimensions.\footnote[2]{
{\bf Note:} Maxwell-independent graphs are called ``$d$-sparse''  and ``$(d, {d+1 \choose 2})$-sparse'' in the literature (see \cite{JacksonBound2011}, \cite{LeeStreinu2008}). But we note that dense and sparse graphs have a variety of different meanings in graph theory. Our terminology is motivated by Maxwell's observation in 1864 that every graph $G$ that is rigid in $d$ dimensions must contain a Maxwell-independent subgraph that has at least $d|V| -{d+1 \choose 2}$ edges \cite{maxwell:equilibrium:1864}.}
Laman's theorem shows that the converse holds for $d=2$ and thus every maximal Maxwell-independent set of $G$ has size equal to the  rank of the $2$-dimensional generic rigidity matroid. While this is
false for $d=3$, we show that every maximal, Maxwell-independent set of a graph $G$ has size at least the rank of the $3$-dimensional generic rigidity matroid of $G$. This answers a question posed by Tib\'or Jord\'an at the 2008 rigidity workshop at BIRS \cite{bib:birs}.

Along the way, we construct subgraphs (1) that yield alternative formulae for a rank upper bound for Maxwell-independent graphs and (2) that contain a maximal (true) independent set. We extend this bound to special classes of non-Maxwell-independent graphs. One further consequence is a simpler proof of
correctness for existing algorithms that give rank bounds.
\end{abstract}

\section{Introduction}\label{sec:intro}
It is a long open problem to combinatorially characterize the $3$-dimensional bar-joint rigidity of graphs.
The problem is at the intersection of combinatorics and
algebraic geometry, and  crops up in practical algorithmic applications ranging from
mechanical computer aided design to
molecular modeling.

\medskip\noindent
The problem is equivalent to combinatorially determining
the generic rank of the $3$-dimensional bar-joint rigidity matrix of a graph.
The {\em $d$-dimensional bar-joint rigidity matrix} of a graph  $G = (V,E)$, denoted $R_d(G)$,
is a matrix of indeterminates. Let $p_1(v), p_2(v), \ldots, p_d(v)$ represent the coordinate position $\p(v) \in \mathbb{R}^d$ of
the {\em joint} corresponding to a
vertex $v\in V$.
The matrix $R_d(G)$ has one row for each edge $e\in E $ and $d$ columns for
each vertex $v \in V$.
The row corresponding to $e = \{u, v\} \in E$
represents the {\em bar}
connecting $p(u)$ to $p(v)$ and has $d$ non-zero indeterminate entries
$\p(u)-\p(v)$ (resp. $\p(v)-\p(u)$), in the $d$ columns corresponding to $u$
(resp. $v$) and zero in the other entries.

\medskip\noindent
A subset of edges $E^\prime$, or a subgraph $(V^\prime, E^\prime)$, of a graph $G=(V, E)$ is said to be {\em independent} (we drop ``bar-joint'' from now on) in $d$-dimensions, when the set of rows of $R_d(G)$ corresponding to $E^\prime$ is {\em generically independent}, or independent for a generic instantiation of the indeterminate entries.
This yields the $d$-dimensional {\em generic rigidity matroid} associated with a graph $G$. The graph is {\em rigid} if the number of generically independent rows or
the rank of $R_d(G)$  is maximal, i.e., $d|V| - {d+1 \choose 2}$, where
${d+1 \choose 2}$ is the number of rotational and translational degrees-of-freedom of a rigid body in $\mathbb{R}^d$ \cite{graver:servatius:rigidityBook:1993}.

\medskip\noindent
Clearly, the number of edges of $G=(V, E)$ is a trivial upper bound on the generic rank of
$R_d(G)$, or alternatively the {\it rank of the $d$-dimensional rigidity matroid of $G$}, which we denote by rank$_d(G)$. Thus, a graph is independent in $d$ dimensions only if $(a)$ $|E|$ does not exceed $d|V| - {d+1\choose 2}$; and $(b)$ this holds for every induced subgraph with at least $d$ vertices. This is called {\em Maxwell's condition in $d$ dimensions} \cite{maxwell:equilibrium:1864}, and we call such graphs (or their edge sets) $G$ {\em Maxwell-independent} in $d$ dimensions. 

\medskip\noindent
In other words, Maxwell's condition states that for any subset of edges of
$G$, independence implies Maxwell-independence. For $d=2$, the famous Laman's theorem states that the converse is also true. I.e., Maxwell-independence implies independence.
Thus (1) the rank of the $2$-dimensional generic rigidity matroid of a graph $G$ is exactly the size of any maximal, Maxwell-independent set (here, by {\em maximal} we mean that no edge can be added without violating Maxwell-independence) and (2) all maximal, Maxwell-independent sets of $G$ must have the same number of edges.

\medskip\noindent
For $d=3$, however, different maximal, Maxwell-independent sets may
have different sizes, see Figure \ref{fig:banana}.
I.e, for $d=3$, the collection of Maxwell-independent sets does
not yield a matroid.
Clearly, any maximal independent subgraph of $G$
is itself Maxwell-independent, so the rank of the generic rigidity matroid of a
graph is at most the
size of {\em some} maximal Maxwell-independent set and this
generalizes to any dimension.
But this only yields the trivial upper bound, i.e., number of edges, for Maxwell-independent graphs.
For other special classes of graphs such as graphs of
bounded degree, graphs that satisfy certain covering conditions
etc., alternative combinatorial formulae are known \cite{JacksonJordanrank:2006, JacksonJordansparse:2005}, that give better bound than the number of edges in some cases.

\begin{center}
\begin{figure}[!h]
\begin{center}
\scalebox{0.3}[0.3]{\includegraphics{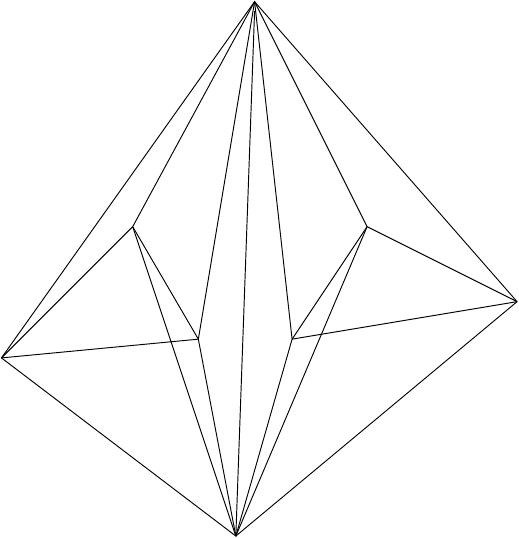} \includegraphics{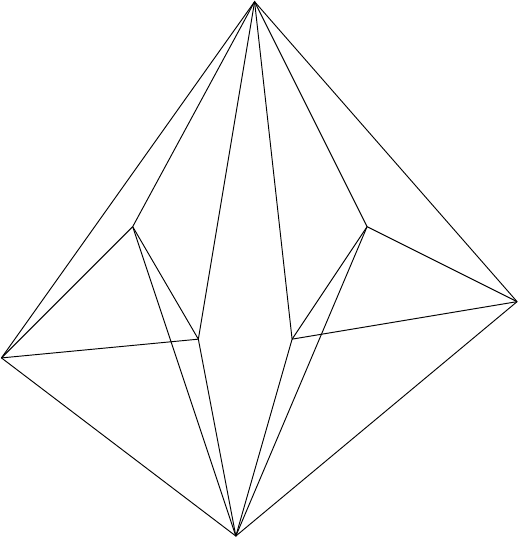} \includegraphics{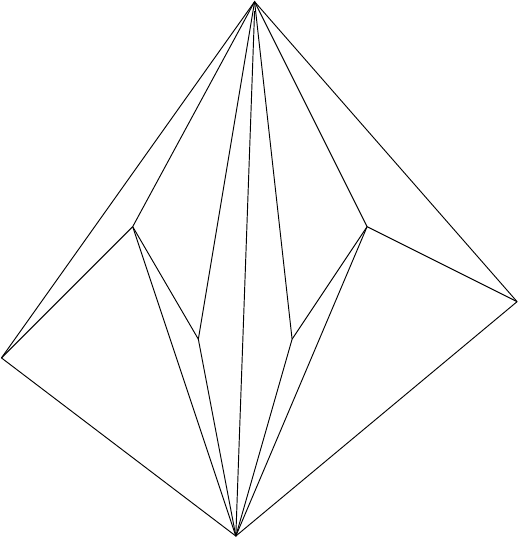}}
\end{center}
\caption{The graph on the left is called a {\em double-banana} and consists of two $K_5$'s intersecting on an edge. The graphs on the middle and the right are two maximal Maxwell-independent sets of different sizes for the graph on the left (the middle is of size $18$ and the right is of size $17$)
.}\label{fig:banana}
\end{figure}
\end{center}
\medskip
\noindent
This leads to the following natural question concerning the rank of the $3$-dimensional generic rigidity matroid. The question was posed by Tib\'or Jord\'an during the 2008 BIRS rigidity workshop \cite{bib:birs}.

\medskip
{\em Question ($\star$):}
Does {\it every} maximal, Maxwell-independent subgraph (subsets of edges) of a graph  $G$ have size at least the rank of the $3$-dimensional generic rigidity matroid of $G$?

Note that the answer to Question ($\star$) would be obvious if every maximal Maxwell-independent set of a given graph $G$ contains a maximal independent set of $G$. However, this is not the case. See Figure \ref{fig:bananaBar}.

\begin{center}
\begin{figure}[!h]
\begin{center}
\scalebox{0.3}[0.3]{\includegraphics{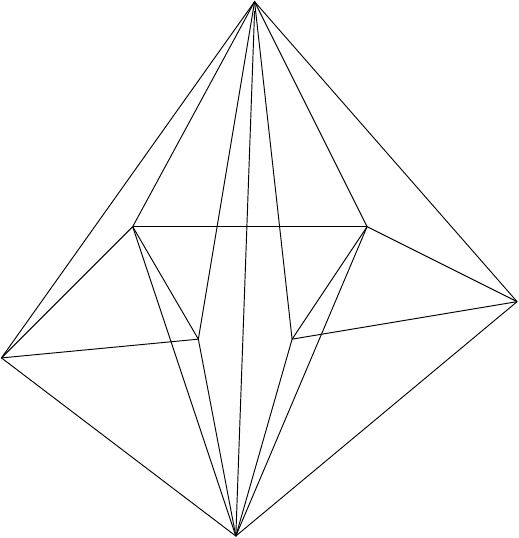}\includegraphics{banana_s1-eps-converted-to.pdf} }
\end{center}
\caption{On the left is a double-banana-bar, which consists of a double-banana and a bar connecting two vertices from each banana. Notice that this double-banana-bar is rigid, thus every maximal independent set in it has $3|V|-6=18$ edges. On the right we have a maximal Maxwell-independent set of the double-banana-bar, which has $3|V|-6=18$ edges. The figure on the right is dependent, so every maximal independent set of it has size less than $3|V|-6=18$. So the right figure cannot contain a maximal independent set of size $3|V|-6$. }\label{fig:bananaBar}
\end{figure}
\end{center}

\medskip
\noindent
Our main result (Theorem \ref{thm:main}) in Section \ref{sec:main} gives an affirmative answer to Question ($\star$) for $d=3$. Bill Jackson \cite{JacksonBound2011} has extended this result up to $d=5$. His proof is by contradiction and is hence nonconstructive. Our proof is constructive: for Maxwell-independent graphs, we give combinatorial formulae based on inclusion-exclusion (IE) counts upper bounding the rank; and we construct subgraphs ({\em independence assignments}) whose sizes meet this bound, and moreover contain a maximal true independent set (Theorems \ref{thm:weakrankIE} and \ref{thm:propermaximal}); this construction is of algorithmic interest. The construction leads to alternative upper bounds on rank related to Dress' formula (\cite{bib:Dress}, Section \ref{sec:knownbounds}) for certain classes of non-Maxwell-independent graphs that admit certain types of covers in Section \ref{sec:nonMaxwell} (Theorems \ref{thm:complete2thin} and \ref{thm:properComplete2thin}). However, algorithms for computing these covers are beyond the scope of this paper.

Several algorithms exist for combinatorially recognizing certain types of dependences for $d=3$  (\cite{bib:survey, andrewThesis,sitharam:zhou:tractableADG:2004}
). The simplest of these algorithms is a
minor modification (\cite{andrewThesis}) of Jacobs and Hendrickson's (\cite{Jacobs97analgorithm}) pebble game for $d=2$, and finds a maximal Maxwell-independent set (it may be neither the minimum sized one nor the 
maximum sized one). The techniques developed in this paper simplify the proofs of correctness for these algorithms.

\medskip

\medskip\noindent
In Section \ref{betterbound}, we also relate our bounds to existing bounds and conjectures.
In the concluding Section \ref{conclusion}, we pose open problems.

\section{Main Result and Proof}\label{sec:main}
In this section, we state and give the proof of the following main theorem. Note that Sections \ref{sec:main} and \ref{betterbound} deal exclusively with $d=3$ and we use rank$(G)$ to denote the rank of the $3$-dimensional generic rigidity matroid of graph $G$.

\begin{thm}\label{thm:main}
Let $\M$ be a maximal Maxwell-independent subgraph of a graph $G=(V,
E)$ and $\I$ be a maximal independent set of the $3$-dimensional generic
rigidity matroid of $G$. Then $|E(\M)| \ge |\I|$, where $E(\M)$ denotes the edge set of $\M$.
\end{thm}

\medskip
\noindent The proof requires a few definitions.

\begin{dfn}\label{dfn:Maxwell}
The {\em Maxwell count} for a graph $G=(V, E)$ in $3$ dimensions is $3|V|-|E|$. $G$ is said to be \emph{Maxwell-rigid} in $3$ dimensions, if there exists a Maxwell-independent subset $E^\star \subseteq E$ such that
the Maxwell count of $G^\star=(V, E^\star)$ is at most $6$. As exceptions, $j$-cliques  ($j\leq 2$) are considered to be Maxwell-independent and Maxwell-rigid.

A subgraph $G^\prime = (V^\prime,E^\prime)$ induced by $V^\prime \subseteq V$ is
said to be a {\em component} of $G$, if it is Maxwell-rigid. In addition, $G^\prime$ is called a {\em vertex-maximal component} of $G$, if it is Maxwell-rigid and there is no proper superset of $V^\prime$ that also induces a Maxwell-rigid subgraph of $G$. A component with $2$ vertices consists of a single edge of the graph, and we call it an {\em edge component}, or {\em trivial} component. Other components are called {\em non-trivial} components.
\end{dfn}


\medskip\noindent
The following concepts of covers and inclusion-exclusion formulae on covers from  \cite{crapo:structuralRigidity:1979, sitharam:zhou:tractableADG:2004, andrewThesis, bib:survey, lovasz:yemini, JacksonJordansparse:2005, JacksonJordanrank:2006} are important for the proof of Theorem \ref{thm:main}.

A \emph{cover} of a graph $G=(V, E)$ is a collection $\X$
of pairwise incomparable induced subgraphs $G_1, \ldots, G_m$ of $G$, each with at least two vertices, such that $\cup_{G_i\in \X} E(G_i) = E$, where $E(G_i)$ is the edge set of subgraph $G_i$. $V(G_i)$ denotes the vertex set of $G_i$. Let $G_i\cup G_j$ denote the graph $(V(G_i)\cup V(G_j)$, $E(G_i)\cup E(G_j))$ and $G_i\cap G_j$ denote the graph $(V(G_i)\cap V(G_j)$, $E(G_i)\cap E(G_j))$.

Given a graph $G$ with a cover $\X$ $=$  $\{G_1, \ldots, G_m\}$, we use $\H(\mathcal{X})$ to denote the set of all pairs of vertices $\{u, v\}$ such that $V(G_i) \cap V(G_j) = \{u,
v\}$ for some $1 \leq i < j \leq m$. Denote by $n_{\{u, v\}}$ the number of elements in $\mathcal{X}$ that contain both $u$ and $v$. Then we can define two different inclusion-exclusion formulae on covers as follows, where the first is used in the proof of Theorem \ref{thm:main} and the second is used later in the paper:

\begin{dfn}\label{dfn:IE}
Given a graph $G=(V, E)$, let $\X =\{e_1, \ldots, e_k$, $G_1, G_2, \ldots, G_m\}$ be a cover of $G$ where $e_1, \ldots, e_k$ are edge components and $G_1, G_2, \ldots, G_m$ are subgraphs with at least $3$ vertices.
The \emph{rank inclusion-exclusion (IE) count} of cover $\mathcal{X}$ is defined as the following:
\begin{equation*}
IE_{\text{rank}}(\X):=\sum\limits_{i=1}^m{rank(G_i)}  -  \sum\limits_{\{u, v\} \in \H(\X)\cap E}(n_{\{u, v\}}
-1)+k
\end{equation*}

The \emph{full rank inclusion-exclusion (IE) count} of cover $\X$ in is defined as
\begin{equation*}
\hbox{IE}_{\text{full}}(\X) :=
\sum\limits_{i=1}^m(3|V(G_i)|-6)  -  \sum\limits_{\{u, v\} \in \H(\X)}(n_{\{u, v\}}
-1)+k
\end{equation*}
\end{dfn}

\medskip
\noindent
The relationships between the two types of IE count defined in Definition \ref{dfn:IE} will be discussed in Section \ref{sec:knownbounds}. The proof of Theorem \ref{thm:main} only uses IE$_\text{rank}$.

Now we are ready to prove Theorem \ref{thm:main}.
\begin{proof}(of Theorem \ref{thm:main})
First, notice that if $\M$ itself is independent, we are done.
Similarly, if $\M$ is Maxwell-rigid, then we have $|E(\M)| = 3|V| -6 \ge$ rank$(G)$ = $|\I|$, hence we are done.

Let $\I_{\M}$ with $|\I_{\M}| =$ rank$(\M)$ be a maximal independent set of $\M$. Without loss of generality, let $\I_{\M} \subseteq \I$. Let $\A := \I\setminus\I_{\M}$. Thus $\text{span}(E(\M))$ $\cap$ $\A= \emptyset$. Here $\text{span}(E(\M))$ means the linear span of those rows of the rigidity matrix $R_3(G)$ corresponding to $E(\M)$.


Consider a cover $\X=\{e_1, \ldots, e_k$ $, \M_1, \M_2, \ldots, \M_m\}$ of $\M$ by the {\em complete} collection of vertex-maximal components, where $e_1, \ldots, e_k$ are edge components and $\M_1,$ $\M_2,$ $\ldots,$  $\M_m$ are non-trivial components. Next we show that for each edge $\{u, v\}$ in $\A$, there exists at least one non-trivial component $\M_i$ such that $u\in \M_i$ and $v\in \M_i$.

Since $\A \cap E(\M)$ $=\emptyset$, $e$ $=\{u, v\}$ $\in \A$ is not an edge component of $\M$. Hence if $u$ and $v$ lie inside any component of $\M$, the component must be non-trivial. If no component $\M_i$ contains both $u$ and $v$, then in fact no vertex-maximal component of $\M$ contains both $u$ and $v$, since $\X$ is the complete collection of vertex-maximal components of $\M$. Next we will show that $\M \cup \{e\}$ is Maxwell-independent.

Suppose not. We know there is a violation to Maxwell's condition in $\M \cup \{e\}$ and this must be caused by the addition of $e$, since $\M$ is Maxwell-independent. To violate Maxwell's condition, both endpoints of $e$ must lie inside a same non-trivial Maxwell-rigid subgraph of $\M$, and every non-trivial Maxwell-rigid subgraph of $\M$ lies inside a non-trivial vertex-maximal component of $\M$. This contradicts the fact that no vertex-maximal component of $\M$ contains both $u$ and $v$. Hence $\M \cup \{e\}$ is Maxwell-independent, contradicting the maximality of $\M$. So for each edge $e=\{u, v\}$ in $\A$, there exists at least one non-trivial component $\M_i$ such that $u\in \M_i$ and $v\in \M_i$.

Denote by $\A_i$ the set of edges of $\A$ both of whose endpoints are in $\M_i$. Hence 
\begin{equation}\label{eqn:IIMDecomp}
|\A| \le \sum_{i=1}^m |\A_i|
\end{equation}


Take $\H(\X)$ and $n_{\{u, v\}}$ as defined earlier in the section. We get
\begin{eqnarray}
|E(\M)|& =& \sum_{i=1}^k 1 +\sum_{i=1}^m |E(\M_i)| - \sum\limits_{\{u, v\} \in \H(\X)\cap E(\M)}(n_{\{u, v\}}
-1) \notag\\
&=& k+ \sum_{i=1}^m |E(\M_i)| - \sum\limits_{\{u, v\} \in \H(\X)\cap E(\M)}(n_{\{u, v\}} 
-1) \label{eqn:Md}
\end{eqnarray}

Since each $\M_i$ is Maxwell-rigid, adding any $e\in \A_i$ into $\M_i$ causes the number of edges in $\M_i$ to exceed $3|V(\M_i)|-6$ and in turn indicates the existence of a true dependence. However, $A_i$ $\cap$ $\text{span}(M_i)$ $=\emptyset$, since $\text{span}(E(\M))$ $\cap$ $\A_i= \emptyset$. It follows that $\M_i$ was already dependent even before $\A_i$ was added. I.e., to obtain an independent set in $\M_i$, at least $|\A_i|$ edges must be removed from $\M_i$. So we have

\begin{equation}\label{eqn:MI}
|E(\M_i)| \geq \text{rank}(\M_i) +|\A_i|
\end{equation}

Plugging \eqref{eqn:MI} into \eqref{eqn:Md}, we have 
 \begin{equation}\label{eqn:sum}
 |E(\M)| \geq \sum_{i=1}^m \text{rank}(\M_i)  - \sum_{\{u, v\} \in \H(\X)\cap E(\M)}(n_{\{u, v\}}
 -1) +\sum_{i=1}^m |\A_i| + k
 \end{equation}

From Proposition \ref{obs:vmmr}(\ref{obs:vmmr1}) below, we know that the cover $\X$ is 2-thin. Then we can apply Theorem \ref{thm:weakrankIE} below and obtain the following:
 \begin{equation}\label{eqn:indep}
 \sum_{i=1}^m \text{rank}(\M_i) - \sum_{\{u, v\} \in \H(\X)\cap E(\M)}(n_{\{u, v\}}-1) +k \geq \text{rank}(\M)=|\I_{\M}|.
 \end{equation}
Then, using \eqref{eqn:sum} and \eqref{eqn:IIMDecomp}, we obtain that 
\begin{eqnarray*}
|E(\M)|& \geq &|\I_{\M}| +\sum_{i=1}^m |\A_i| \hbox{ (using \eqref{eqn:sum} and \eqref{eqn:indep})}\\
 &\geq &|\I_{\M}| + |\A|  \hbox{ (using \eqref{eqn:IIMDecomp})} \\
 &=&|\I|,
\end{eqnarray*} which proves Theorem \ref{thm:main}. 
\end{proof}

\medskip\noindent
Note that the proof of Theorem \ref{thm:main} uses a cover by the complete collection of vertex-maximal components. This not only implies 2-thinness of the cover, but also strong 2-thinness. However, 2-thinness (Proposition \ref{obs:vmmr}(\ref{obs:vmmr1})) is sufficient for proving Theorem \ref{thm:main}. Strong 2-thinness is used in Section \ref{betterbound}.


In the remainder of this section, we state and prove Proposition \ref{obs:vmmr} and Theorem \ref{thm:weakrankIE} and the required lemmas. The following concept, as defined in \cite{JacksonJordanrank:2006}, is needed to state Proposition \ref{obs:vmmr}.

\begin{dfn}\label{dfn:2thin}
Let $\X =\{G_1, G_2, \ldots, G_m\}$ be a cover $G$. We say $\X$ is {\em $2$-thin} if $| V(G_i) \cap V(G_j)| \leq 2$ for all $1 \leq i < j \leq m$.  We say a $2$-thin cover $\X$ is {\em strong 2-thin} if for all $1 \leq i < j \leq m$, whenever $|V(G_i) \cap V(G_j)|=2$, then $G_i$ and $G_j$ in fact share an edge.
\end{dfn}

Next, we prove a lemma illustrating an elementary, but useful property of the union of two Maxwell-rigid graphs.
\begin{lem}\label{lem:join}
\begin{enumerate}[(a)]
\item \label{lem:join1} Given Maxwell-rigid graphs $\M_1$ and $\M_2$, if $V(\M_1)$ $\cap$ $V(\M_2)$ consists of two vertices $u$ and $v$ and $\{u, v\}$ $\not\in$ $E(\M_1)$ $\cup$ $E(\M_2)$, then $\M_1$ $\cup$ $\M_2$ is also Maxwell-rigid. 
\item \label{lem:join2} Given Maxwell-independent graph $\M$ and two Maxwell-rigid subgraph $\M_1$ and $\M_2$ of $\M$, if $|V(\M_1) \cap V(\M_2)|$ $\ge$ $3$, then $\M_1$ $\cup$ $\M_2$ is also Maxwell-rigid.
\end{enumerate}

\end{lem}
\begin{proof}
\begin{enumerate}[(a)]
\item Let $\N_1$ be a Maxwell-independent subgraph of $\M_1$ with $3|V(\M_1)$ $-$ $6|$ edges and $\N_2$ be a Maxwell-independent subgraph of $\M_2$ with $3|V(\M_2)$ $-$ $6|$ edges. We show next that $\N_1$ $\cup$ $\N_2$ is Maxwell-independent.

 Suppose $\N_1$ $\cup$ $\N_2$ is Maxwell-dependent. Then there exists $\N^\prime \subseteq$ $\N_1$ $\cup$ $\N_2$ such that $\N^\prime$ has Maxwell count less than $6$. Since both $\N_1$ and $\N_2$ are Maxwell-independent, it is clear that $\N^\prime \nsubseteq \N_1$ and $\N^\prime \nsubseteq \N_2$. Let $\N^\prime=\N_1^\prime \cup \N_2^\prime$ such that $\N_1^\prime \subseteq \N_1$ and $\N_2^\prime \subseteq \N_2$. Then $\N_1^\prime$ and $\N_2^\prime$ both have Maxwell count at least 6. To make their union have Maxwell count less than 6, $\N_1^\prime$ and $\N_2^\prime$ must share at least two vertices. Since $V(\N_1) \cap V(\N_2)$ consists of two vertices $u$ and $v$, we know $V(\N_1^\prime)$ $\cap$ $V(\N_2^\prime)$ consists of at most two vertices $u$ and $v$.

Since $\{u, v\}\notin E(\N_1)\cup E(\N_2)$, it can be seen that in order to make $\N^\prime$ of Maxwell count less than $6$, at least one of $\N_1^\prime$ and $\N_2^\prime$ will have Maxwell count less than $6$, which together with the fact that $\N_1^\prime \subseteq \N_1$ and $\N_2^\prime \subseteq \N_2$ violates Maxwell-independence of $\N_1$ or $\N_2$. Hence $\N_1$ $\cup$ $\N_2$ is Maxwell-independent. Notice that $\N_1$ $\cup$ $\N_2$ has enough edges to be Maxwell-rigid and thus $\M_1$ $\cup$ $\M_2$ is also Maxwell-rigid.

\item Since $\M$ is Maxwell-independent, we know both $\M_1 \cup \M_2$  and $\M_1 \cap \M_2$ are Maxwell-independent. Then can we calculate the Maxwell count of $\M_1 \cup \M_2$ as follows. We know (1) $\M_1$ and $\M_2$ each have Maxwell count $6$ and (2) $\M_1$ $\cap$ $\M_2$ has Maxwell count at least $6$ since $\M_1$ $\cap$ $\M_2$ is Maxwell-independent and has at least $3$ vertices. Thus the Maxwell count of $\M_1$ $\cup$ $\M_2$ is at most $6$. Together with the fact that $\M_1$ $\cup$ $\M_2$ is Maxwell-independent, we know $\M_1$ $\cup$ $\M_2$ is Maxwell-rigid.
\end{enumerate}
\end{proof}

This following proposition gives a useful property of a cover of a Maxwell-independent subgraph by vertex-maximal components.
\begin{obs}\label{obs:vmmr}
Let $\M$ be a Maxwell-independent graph. Let $\X=\{e_1, \ldots, e_k,$ $\M_1, \M_2, \ldots, \M_m\}$ be a cover of $\M$ by vertex-maximal components, where $e_1, \ldots, e_k$ are edge components and $\M_1, \M_2, \ldots, \M_m$ are non-trivial components. Then
\begin{enumerate}[(a)]
\item\label{obs:vmmr1} $\X$ is a 2-thin cover of $\M$.
\item\label{obs:vmmr2} $\X$ is strong 2-thin.

\end{enumerate}
\end{obs}
\begin{proof}
\begin{enumerate}[(a)]
\item Edge components do not affect the cover being 2-thin or not. 
Let $\M_i$ and $\M_j$ be two non-trivial vertex-maximal components in $\M$. Suppose $\M_i$ and $\M_j$ share at least $3$ vertices. Then from Lemma \ref{lem:join}(\ref{lem:join2}), we know $\M_i$ $\cup$ $\M_j$ is Maxwell-rigid, violating the fact that $\M_i$ and $\M_j$ are vertex-maximal components.

\item  Again, edge components do not affect the cover being strong 2-thin or not. Let $\M_i$ and $\M_j$ be two non-trival vertex-maximal components in $\M$. If $\M_i$ and $\M_j$ share two vertices but do not share an edge, then from Lemma \ref{lem:join}(\ref{lem:join1}), $\M_i$ $\cup$ $\M_j$ is Maxwell-rigid, which violates the vertex-maximal property of $\M_i$ and $\M_j$.



\end{enumerate}
\end{proof}

\medskip\noindent
Next we prove a lemma about the structure of a 2-thin cover of a Maxwell-independent graph.
We first need the following definition of {\em 2-thin component graph}.

\begin{dfn}\label{dfn:componentGraph}
Given graph $G=(V, E)$, let $\X=$ $\{G_1, G_2,$ $ \ldots,$ $ G_m\}$ be a 2-thin cover of $G$ by components of $G$. The {\em 2-thin component graph} $\C_{\X}$ of $G$ ({\em component graph} for short) is defined as follows. $V(\C_{\X}) $ $=$ $V_\text{component}(\C_{\X})$ $\cup $ $V_\text{edge}(\C_{\X})$, where $V_\text{component}(\C_{\X})$ consists of {\em component nodes} $C_{G_i}$, one for each component $G_i$ in $\X$; and $V_\text{edge}(\C_{\X})$ consists of {\em edge nodes} $C_e$, one for each edge $e$ shared by at least two components in $\X$. The edges in $E(\C_{\X})$ are of the form $(C_{G_i}, C_e)$, where $C_{G_i}$ $\in$ $V_\text{component}(\C_{\X})$, $C_e$ $\in$ $V_\text{edge}(\C_{\X})$, and $e\in E$ is a shared edge of $G_i$.
\end{dfn}
\noindent
Figure \ref{compG} shows how to obtain a 2-thin component graph from a graph and a cover by its vertex-maximal components.

\medskip
\noindent
Note that components sharing only vertices are non-adjacent in the component graph. Edge components have degree zero and become disconnected nodes in the component graph. See Figure \ref{compG}.

\begin{center}
\begin{figure}[!h]
\begin{center}
\includegraphics[width=.5\textwidth]{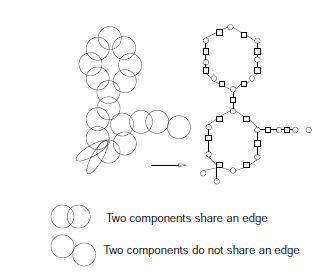}
\end{center}
\caption{The figure on the left represents the vertex-maximal components of a graph. On the right side is its 2-thin component graph, where circles represent component nodes and squares represent edge nodes. Note that the 2-thin component graph may not be connected.}\label{compG}
\end{figure}
\end{center}

\noindent
Lemma \ref{lem:comb}(\ref{thm:9tree}) below states an important property of component graphs of Maxwell-independent graphs. Specifically,  these 2-thin component graphs generalize the concept of partial $m$-trees (also called {\em tree-width $m$ graphs}) and Henneberg constructions \cite{graver:servatius:rigidityBook:1993}, which we define below.

\medskip
\noindent
\begin{dfn}\label{dfn:comppartialmTree}
Let $m$ be a positive integer. Then a 2-thin component graph is called a {\em generalized partial $m$-tree} if it can be reduced to an empty graph by a sequence of the 
following two operations: (i) removal of a component node of degree at most $m$ and 
(ii) removal of an edge node of degree one. 
\end{dfn}

Now we are ready to state the lemma. 
\begin{lem}\label{lem:comb} If $\M$ is a Maxwell-independent graph and $\X$ is a 2-thin cover of $\M$ by components of $\M$, then
\begin{enumerate}[(a)]
\item \label{thm:degree}
the component nodes of any subgraph of the 2-thin component graph $\C_{\X}$ have average degree strictly less than $4$.
 \item \label{thm:9tree}
any subgraph of the 2-thin component graph $\C_{\X}$ of $\M$ is a generalized partial $3$-tree.
 \end{enumerate}
\end{lem}
\begin{proof}
\begin{enumerate}[(a)]
\item\label{lem:comb1} First we remove all edge components of $\M$ and show the remainder of the component graph has average degree $<4$.

Let $\K_{\X}$ be any subgraph of the 2-thin component graph $\C_{\X}$. Let $K$ denote $\K_{\X}$'s corresponding subgraph in $\M$.
Let $\{\M_1, \ldots, \M_n\}$ be $\X$ restricted to $K$. Let $V_i$ and $E_i$ be the shared vertex and shared edge sets of component $\M_i$ of $K$, i.e., $V_i$ and $E_i$ are shared by other components $\M_j$ of $K$. Let $V_s$ and $E_s$ be the entire sets of such shared vertices and shared edges in $K$. 
Let $n_e$ and $n_v$ denote the number of components $\M_i$ of $K$ that share $e$ and $v$ respectively. Since the Maxwell count of each $\M_i$ is $6$ (they are all non-trivial), the Maxwell count of $K$ can be calculated as follows:
\begin{equation*}
\sum_{i} 6 - 3\sum\limits_{v\in V_s}n_v +\sum\limits_{e\in E_s}n_e + 3|V_s|-|E_s|=\sum_{i}(6- 3|V_i|+|E_i|)+ 3|V_s|-|E_s|
\end{equation*}

Suppose the Maxwell count of $K$ is $\geq 6$. We have
\begin{equation}\label{eqn:MC}
6n - 6 \geq 3\sum\limits_{i} |V_i| -\sum\limits_{i} |E_i| -3|V_s|+|E_s|
\end{equation}

Consider any shared vertex $v$ in $V_s$. Denote by $C_v \subseteq \{1,\ldots,n\}$ the set of indices of components containing $v$. In this proof, since the context is clear, we refer to $\M_j$, $j\in C_v$ as a component containing $v$. The collection of all $n_v$ components of $K$ meeting at $v$ forms a subgraph $C$. Since $K$ is Maxwell-independent, $C$ is also Maxwell-independent. Let $w_v^{j}$ be the number of shared edges incident at $v$ in component $\M_j$ and $s_v$ be the number of shared edges that are incident at $v$. Then the Maxwell count of $C$ can be computed as follows:
\begin{itemize}
\item there are $n_v$ components, which contributes $6 n_v$;
\item $v$ is shared by $n_v$ components, and the contribution is $ - (3n_v -3)$;
\item each shared edge in a component $\M_{j}$ contributes $1$ to the Maxwell count, and altogether the shared edges contribute $(\sum_{j \in  C_v} w_v^{j}) - s_v$
\item for each shared edge $e=\{u, v\}$, vertex $u$ contributes $ -3[(\sum_{j \in  C_v} w_v^{j}) - s_v]$
\item for the set of shared vertices that are not part of any shared edge in $C$, their contribution is $-\Delta$ for a non-negative number $\Delta$;
\end{itemize}
Thus the Maxwell count of $C$ is:
\begin{equation*}
3n_v - 2[(\sum\limits_{j \in  C_v} w_v^{j}) - s_v] +3 - \Delta
\end{equation*}
Since $C$ is Maxwell-independent, we know:
\begin{eqnarray*}
3n_v - 2 [(\sum\limits_{j \in  C_v} w_v^{j}) - s_v] +3 - \Delta& \geq& 6\\
\end{eqnarray*}
Since $\Delta $ $\ge 0$, we know
\begin{eqnarray*}
3n_v - 2 [(\sum\limits_{j \in  C_v} w_v^{j}) - s_v] &\ge& 3
\end{eqnarray*}
Summing over all shared vertices in $V_s$, we have:
\begin{equation*}
3\sum\limits_{v\in V_s}n_v - 2 \sum\limits_{v\in V_s}[(\sum\limits_{j \in  C_v} w_v^{j}) - s_v] \geq 3|V_s|
\end{equation*}
Since $\sum\limits_{v\in V_s}n_v  =\sum\limits_i |V_i|$, $\sum\limits_{v\in V_s}(\sum\limits_{j \in  C_v} w_v^{j})= 2\sum\limits_i |E_i|$ and $\sum\limits_{v\in V_s} s_v =2 |E_s|$, we know
\begin{equation*}
3\sum\limits_{i} |V_i| -4 \sum\limits_{i} |E_i| -3|V_s|+ 4|E_s| \geq 0
\end{equation*}
Plugging into \eqref{eqn:MC}, we have:
\begin{eqnarray*}
6n - 6 &\geq& 3\sum\limits_{i} |V_i| -\sum_{i} |E_i| -3|V_s|+|E_s|\\
&\geq &3\sum\limits_{i} |V_i| -4 \sum_{i} |E_i| -3|V_s|+ 4|E_s|\\
&& + 3( \sum\limits_{i} |E_i| -|E_s|)\\
&\geq&3( \sum\limits_{i} |E_i| -|E_s|)
\end{eqnarray*}
Since $|E_s| \leq \frac{1}{2}\sum\limits_{i} |E_i|$, we have:
\begin{equation*}
6n - 6 \geq \frac{3}{2} \sum\limits_{i} |E_i|
\end{equation*}
We now observe that the component nodes in $\K_{\X}$ must have average degree strictly less than $4$. Otherwise, $\sum_{i} |E_i| \geq 4n$, leading to a contradiction that
\begin{equation*}
6n - 6 \geq \frac{3}{2} 4n = 6n.
\end{equation*}
This proves (\ref{lem:comb1}).

\item This follows immediately from (\ref{lem:comb1}).

\end{enumerate}

\end{proof}

Next we establish a condition on the cover of a Maxwell-independent graph such that the IE$_\text{rank}$ count in Definition \ref{dfn:IE} gives an upper bound on rank$(\M)$. This condition is called an {\em independence assignment}.

\begin{dfn}\label{dfn:assign}
Given a graph $G=(V, E)$ and a cover $\X=\{G_1, \ldots, G_m\}$ of $G$, we say $(G, \X)$ has an {\em independence assignment} $[\I$; $\{\I_1, \ldots,\I_m\}]$, if there is an independent set $\I$ of $G$ and maximal independent set
$\I_i$ of each of the $G_i$'s, such that $\I$ {\em restricted to} $G_i$, (denoted $\I|_i$), is contained in $\I_i$ and for any $e\in \H(\X)$, $e$
is missing from at most one of the $\I_i$'s whose corresponding $G_i$ contains $e$. When $\X$ is clear, we also say there is an independence assignment for $G$. 
\end{dfn}

\medskip
\noindent
The next lemma shows the existence of an independence assignment for Maxwell-independent graphs.

\begin{lem}\label{lem:assign}
If $\M$ is Maxwell-independent and $\X$ is a 2-thin cover of $\M$ by components of $\M$, then $(\M, \X)$ has an independence assignment.
\end{lem}
\begin{proof} (of Lemma \ref{lem:assign}).
In fact, we can construct an independence assignment if the 2-thin component graph of $\M$ is a generalized partial $9$-tree.
From Lemma \ref{lem:comb}(\ref{thm:9tree}), we know that any subgraph of the 2-thin component graph $\C_{\X}$ of $\M$ is a generalized partial
$3$-tree, which is automatically a generalized partial $9$-tree. Let $\M_1, \M_2, \ldots$ $\M_n$ be the component nodes of $\M$ listed in reverse order from the removal order in Definition \ref{dfn:comppartialmTree}. We use induction to prove that there is always an independence assignment for $(\M, \X)$.

If $\X$ has only one component, it is clear that we can find an independence assignment.

Suppose there is an independence assignment $[\I^{k}$; $\I^{k}_i$ $1\leq i\leq k]$ for a subgraph $\C_{\X}^k$ of $\C_{\X}$ containing the component nodes $\M_1, \M_2, \ldots, \M_k$. After adding $\M_{k+1}$ to form $\C_{\X}^{k+1}$, we need to find $\I^{k+1}$, which is a maximal independent set of $\bigcup\limits_{i=1}^{k+1}\M_i$, and $\I^{k+1}_i$ for $1\le i\le k+1$ such that $[\I^{k+1}$; $\I^{k+1}_i$ $1\leq i\leq k+1]$ is an independence assignment.

First we take $\I^{k+1}_i := \I^{k}_i$ for $1 \leq i \leq k$ and let $\S$ be the set of edges of $\M_{k+1}$ that are shared by other components. Since $|\S| \leq 9$, $\S$ is independent for $d=3$, because for $d=3$, a minimum-size graph that is not independent will have at least $10$ edges. Thus we can extend $\S$ to a maximal independent set $\I^{k+1}_{k+1}$ of $\M_{k+1}$. Now let $\I^{k+1}$ $:=\I^{k}$ $\cup$ $\I^{k+1}_{k+1}$, then (1) $\I^{k+1}$ spans all edges in $\bigcup\limits_{i=1}^{k+1}\M_i$, and (2) every edge $e$ in $\I^{k+1}$ that is shared by at least two components in $\M_1, \M_2, \ldots$ $\M_{k+1}$ is missing in at most $1$ of the $\I^{k+1}_i$'s sharing $e$, since (a) $[\I^{k}$; $\I^{k}_i$ $1\leq i\leq k]$ is an independence assignment for $\C_{\X}^k$ and (b) $\I^{k+1}_{k+1}$ contains all shared edges of $\M_{k+1}$. If $\I^{k+1}$ is already independent, we have our independence assignment. Otherwise we can remove a minimum number of edges from $\I^{k+1}$ until it is independent.

\end{proof}

The following theorem gives an alternative combinatorial upper bound on rank of rigidity matroid of Maxwell-independent graphs. This also completes the proof for Theorem \ref{thm:main}.

\begin{thm}\label{thm:weakrankIE}
Let $\M$ be a Maxwell-independent graph and $\X=\{e_1, \ldots, e_k$, $\M_1, \M_2, \ldots, \M_m\}$ be a 2-thin cover of $\M$ by components of $\M$. Then $\sum_{i=1}^m rank(\M_i) - \sum_{\{u, v\} \in \H(\X)\cap E(\M)}(n_{\{u, v\}}-1) +k \geq \text{rank}(\M)$.

\end{thm}
\begin{proof}

When $\X$ is a 2-thin cover, we can apply Lemma \ref{lem:assign} and obtain that $(\M, \X)$ has an independence assignment.

First we remove all edge components of $\M$ to obtain a new graph $\M^\prime$. Now the existence of an independence assignment directly implies that $\sum_{i=1}^m$ $rank(\M_i)$ $-$ $\sum_{\{u, v\} \in \H(\X)\cap E(\M)}(n_{\{u, v\}}-1)$ $\geq$ $\text{rank}(\M^\prime)$.

Next we consider the edge components $e_1, \ldots, e_k$. If we add the contributions of all of them to both sides of the inequality, the left hand side becomes $\sum_{i=1}^m rank(\M_i) - \sum_{\{u, v\} \in \H(\X)\cap E(\M)}(n_{\{u, v\}}-1) +k$, and the right hand side becomes  $ |\I_{\M^\prime}| + k$, which is at least the rank of $\M$, since $E(\M)$ $=$ $E(\M^\prime)$ $\cup\{e_1, \ldots, e_k\}$.

\end{proof}

\section{Alternative Upper Bounds Using IE Counts}\label{betterbound}
\subsection{Relation to Known Bounds and Conjectures Using IE Counts}\label{sec:knownbounds}
\noindent Decomposition of graphs into covers is a natural way of approaching a combinatorial characterization of $3$-dimensional rigidity. So far, the
inclusion-exclusion(IE) count method for covers has been used by many in the
literature (see \cite{crapo:structuralRigidity:1979, sitharam:zhou:tractableADG:2004, andrewThesis, bib:survey, lovasz:yemini, JacksonJordansparse:2005, JacksonJordanrank:2006}). The most explored decompositions are the 2-thin covers.

\medskip\noindent
We defined two types of rank IE counts in Definition \ref{dfn:IE}, with IE$_\text{rank}$ being used in the proof of Theorem \ref{thm:main}. Our Theorem \ref{thm:propermaximal} below in Section \ref{sec:proper}, will show that for a {\em specific, not necessarily independent} cover, a slightly different inclusion-exclusion count is equal to IE$_\text{rank}$ count, which in turn gives a rank upper bound for Maxwell-independent graphs.

\medskip\noindent
Besides the IE$_\text{rank}$ count, other IE counts have also been explored in the aforementioned literature.
In 1983, Dress et al \cite{bib:Dress,bib:Tay84} conjectured that the minimum of the IE$_\text{full}$ count taken over all 2-thin covers is an upper bound on the rank of the 3-dimensional generic rigidity matroid. However,
this conjecture was disproved for general graphs by Jackson and Jord\'an in
\cite{Jackson03thedress}. 

\medskip\noindent
Although Dress' conjecture is false, the IE$_\text{full}$ count can be an upper bound of the rank if the cover is special: it is shown in \cite{JacksonJordanrank:2006} that the minimum of the IE$_\text{full}$ count taken over all {\em independent} 2-thin covers is an upper bound on the rank. Here, an {\em independent 2-thin cover} $\X$ is one for which the edge set given by the pairs in the shared part $\H(\X)$ is independent. It is also shown that to achieve the upper bound, the covers need not be independent, but can be obtained as iterated, or recursive version of independent covers. 

\medskip\noindent
We have no examples where our bound in Theorem \ref{thm:propermaximal} is better than the above mentioned bound from \cite{JacksonJordanrank:2006}, which was conjectured to be tight when restricted to non-rigid graphs and covers of size at least $2$. Hence any such examples would be counterexamples to their conjecture.
However, our formula provides an alternative way of computing a rank upper bound using not necessarily independent covers.

\medskip\noindent
In Section \ref{sec:nonMaxwell}, we use the same IE$_{\text{full}}$ count
over another special cover, which is a {\em specific non-iterated, non-independent} cover, to obtain rank bounds on Maxwell-dependent graphs. 
Again, we have no examples where our
bound is better than the above mentioned bound in \cite{JacksonJordanrank:2006}, which was conjectured to be tight. Hence any such examples would be counterexamples to their
conjecture. Our bound gives an alternative method using a specific,
non-iterated, not necessarily independent cover by (proper) vertex-maximal components. However, the catch is that these covers may not exist for general graphs.

\subsection{Alternative Upper Bounds for Maxwell-Independent Graphs}\label{sec:proper}
In this section, we give alternative combinatorial bounds on the rank of the generic rigidity matroid of Maxwell-independent graphs in $3$ dimensions.

Notice that if $\M$ is a Maxwell-independent graph with a cover $\X$ $=$ $\{e_1, \ldots, e_k,$ $\M_1, \M_2, \ldots, \M_m\}$ by vertex-maximal components, then $\H(\X)$ $=$ $\H(\X) \cap E(\M)$ and thus $\sum_{i=1}^m$ $rank(\M_i)$ $-$ $\sum_{{\{u, v\}} \in \H(\X)}$ $(n_{\{u, v\}} -1)$ $+k$ $=$ IE$_\text{rank}(\X)$ $\ge$ rank$(\M)$.

However, when a graph $\M$ is Maxwell-rigid, there is a single vertex-maximal component namely $\M$ itself, so the above bound is uninteresting. In this case, we use the cover of $\M$ by ``proper'' vertex-maximal components:

\begin{dfn}\label{dfn:properMax}
Given graph $G=(V, E)$, an induced subgraph is \emph{proper vertex-maximal, Maxwell-rigid} if it is Maxwell-rigid and the only graph that properly contains this
subgraph and is Maxwell-rigid is $G$ itself.
\end{dfn}

\noindent Since the collection of proper vertex-maximal components may not be a 2-thin cover even for Maxwell-independent graphs, Theorem \ref{thm:weakrankIE} does not directly apply. The following theorem deals with cases that are relatively minor variations of Theorem \ref{thm:weakrankIE}.

\begin{thm}\label{thm:propermaximal}
Let $\M$ be a Maxwell-independent graph and $\X$ $=$ $\{e_1, \ldots, e_k,$ $\M_1, \M_2, \ldots, \M_m\}$ be a cover of $\M$ by proper vertex-maximal components. Then we have:
\begin{enumerate}
\item If $\X$ is strong $2$-thin, then $\sum_{i=1}^m rank(\M_i)$ $-$ $\sum_{{\{u, v\}} \in \H(\X)}$ $(n_{\{u, v\}} -1)$ $+k$ $=$ IE$_\text{rank}(\X)$ $\ge$ rank$(\M)$.
\item If $\X$ is $2$-thin but not strong $2$-thin, $\X$ consists entirely of two non-trivial components $\M_i$ and $\M_j$ in $\X$ s.t. $\M$ $=$ $\M_i$ $\cup$ $\M_j$ and hence $\text{rank}(\M_i) +\text{rank}(\M_j)$ $-$ $\text{rank}(\M_i$ $\cap$ $\M_j)$ $\geq$ rank$(\M)$.
\item Otherwise, there exist two non-trivial components $\M_i$ and $\M_j$ in $\X$, s.t. $\M$ $=$ $\M_i$ $\cup$ $\M_j$ and hence $\text{rank}(\M_i) +\text{rank}(\M_j)$ $-$ $\text{rank}(\M_i$ $\cap$ $\M_j)$ $\geq$ rank$(\M)$.
\end{enumerate}

\end{thm}
\begin{proof}
\begin{enumerate}
\item When $\X$ is strong $2$-thin, we know $\H(\X)$ $=$ $\H(\X) \cap E(\M)$ and thus $\sum_{i=1}^m rank(\M_i)$ $-$ $\sum_{{\{u, v\}} \in \H(\X)}$ $(n_{\{u, v\}} -1)$ $+k$ $=$ IE$_\text{rank}(\X)$. Then it follows from Theorem \ref{thm:weakrankIE} that IE$_\text{rank}(\X)$ $\geq$ rank$(\M)$.

\item When $\X$ is $2$-thin but not strong $2$-thin, we know there exist two proper vertex-maximal components $\M_i$ and $\M_j$, s.t. $\M_i$ $\cap$ $\M_j$ has two vertices but no edge. From Lemma \ref{lem:join}(\ref{lem:join1}), we know $\M_i$ $\cup$ $\M_j$ is Maxwell-rigid. Since $\M_i$ and $\M_j$ are both proper vertex-maximal, we know $V(\M)$ $= $ $V(\M_i)$ $\cup$ $V(\M_j)$. Since $\M$ is Maxwell-independent, we know $E(\M)$ $= $ $E(\M_i)$ $\cup$ $E(\M_j)$. Since the cover is $2$-thin, no other non-trivial vertex-maximal component can exist. Hence $\M_i$ and $\M_j$ are the only two non-trivial components in $\X$ and it follows that rank$(\M_i$ $\cap$ $\M_j)$ $=0$ and hence $\text{rank}(\M_i) +\text{rank}(\M_j)$ $- \text{rank}$ $(\M_i$ $\cap$ $\M_j)$  $\geq$ rank$(\M)$.

\item When $\X$ is not $2$-thin, i.e., there exist $\M_i$ and $\M_j$ such that their intersection has at least $3$ vertices. From Lemma \ref{lem:join}(\ref{lem:join2}), we know the union of $\M_i$ and $\M_j$ is also Maxwell-rigid. Since $\M_i$ and $\M_j$ are both proper vertex-maximal, we know $V(\M)$ $= $ $V(\M_i)$ $\cup$ $V(\M_j)$. Since $\M$ is Maxwell-independent, we know $E(\M)$ $= $ $E(\M_i)$ $\cup$ $E(\M_j)$.

It remains to show that $\text{rank}(\M_i) +\text{rank}(\M_j)$ $- \text{rank}$ $(\M_i$ $\cap$ $\M_j)$ $\geq$ rank$(\M_i\cup \M_j)$. To show this, we can start from a maximal independent set $\I$ of $\M_i$ $\cap$ $\M_j$, and expand it to maximal independent sets $\I_i$ of $\M_i$ and $\I_j$ of $\M_j$. It is clear that $\I_i \cup \I_j$ spans the graph $\M_i$ $\cup$ $\M_j$, and hence $\text{rank}(\M_i) +\text{rank}(\M_j)$ $- \text{rank}$ $(\M_i$ $\cap$ $\M_j)$ $=$ $|\I_i \cup \I_j|$ $\geq$ rank$(\M_i\cup \M_j)$.

\end{enumerate}

\end{proof}

\subsection{Removing the Maxwell-Independence Condition}\label{sec:nonMaxwell}
We now give rank bounds for Maxwell-dependent graphs using the IE$_{\text{full}}$ count. We start with the following simple but useful property of edge-sharing, Maxwell-rigid subgraphs.

\begin{lem}\label{obs:maximal}
Given graph $G=(V, E)$, let $G_1$ and $G_2$ be two subgraphs of $G$ s.t. $G_0$ $=$ $G_1\cap G_2$ consists of two vertices $u$, $v$ and an edge $e=\{u, v\}$.
\begin{enumerate}[(a)]
\item\label{obs:maximal1} If $G_1$ is a vertex-maximal component of $G$ and there is a  Maxwell-independent subgraph $\M_1$ of $G_1$ s.t. $|E(\M_1)|$ $=$ $3|V(\M_1)|-6$ and $e\not\in \M_1$, then every maximal Maxwell-independent subgraph of $G_2$ contains $e$.
\item\label{obs:maximal2} If $G_1$ is a proper vertex-maximal component of $G$ and there is a  Maxwell-independent subgraph $\M_1$ of $G_1$ s.t. $|E(\M_1)|$ $=$ $3|V(\M_1)|-6$ and $e\not\in \M_1$, then one of following holds:
(1) $V$ $=$ $V(G_1)$ $\cup$ $V(G_2)$, or
(2) every maximal Maxwell-independent subgraph of $G_2$ contains $e$.
\end{enumerate}
\end{lem}
\begin{proof}
\begin{enumerate}[(a)]
\item Suppose there is a Maxwell-independent subgraph $\M_2$ of $G_2$
such that $\{e\}$ $\cup$ $E(\M_2)$ is Maxwell-dependent. Then there must be a subgraph $\M_2^{\prime}$ of $\M_2$ such that $\M_2^{\prime}$ has Maxwell count $6$. Then it follows from Lemma \ref{lem:join}(\ref{lem:join1}) that $\M_2^{\prime}$ $\cup G_1$ is also Maxwell-rigid, a contradiction to the vertex-maximality of $G_1$.
\item Statement follows from (\ref{obs:maximal1}) and the proper vertex-maximality of $G_1$. 
\end{enumerate}
\end{proof}

Next we give two similar theorems with similar proofs. The first theorem, Theorem \ref{thm:complete2thin}, gives a rank bound for graphs for which the complete collection of vertex-maximal components forms a $2$-thin cover. The second, Theorem \ref{thm:properComplete2thin}, concerns proper vertex-maximal components.

\begin{thm}\label{thm:complete2thin}
For a graph $G=(V, E)$, if the complete collection $\X$ $=$ $\{e_1, \ldots,$ $e_k$, $G_1, G_2, \ldots, G_m\}$ of vertex-maximal components forms a 2-thin cover, then IE$_{\text{full}}(\X)$ is an upper bound on $\text{rank}(G)$, i.e., $$\sum\limits_{i=1}^m{(3|V(G_i)| - 6)} - \sum\limits_{\{u, v\} \in \H(\X)}(n_{\{u, v\}}
-1) + k \geq \text{rank}(G).$$ 
\end{thm}

\begin{proof}
We first consider the case where there are no edge components.

First, we show that the cover $\X$ is strong 2-thin. Suppose not, then there exists $\{u, v\}\in \H(\X)$ such that $\{u, v\}\notin E$. Suppose further that $G_i$ and $G_j$ both contain $u$ and $v$. From Lemma \ref{lem:join}(\ref{lem:join1}), we know $G_i\cup G_j$ is Maxwell-rigid, contradicting the fact that $G_i$ and $G_j$ are vertex-maximal, Maxwell-rigid.  Hence the cover $\X$ is strong 2-thin and IE$_{\text{full}}(\X)$ can be rewritten as $\sum\limits_{i=1}^m (3|V(G_i)| - 6)
-\sum\limits_{\{u, v\} \in \H(\X)\cap E}(n_{\{u, v\}}
-1)$.

We need the following claim (which is also used for proving Theorem \ref{thm:properComplete2thin}).

\medskip
\begin{clm}\label{clm:2thin}
For a graph $G$ $=$ $(V, E)$, if the complete collection $\X$ $=$ $\{G_1$, $G_2$, $\ldots$, $G_m\}$ of (proper) vertex-maximal components forms a strong 2-thin cover, then there is a maximal Maxwell-independent subgraph $\M$ of $G$ s.t. IE$_{\text{full}}(\X)$ $=$ $|E(\M)|$ and hence IE$_{\text{full}}(\X)$ $\ge$ $\text{rank}(G)$.
\end{clm}

\begin{proof}
We show the claim for the case where $\X$ consists of vertex-maximal components. However, along the way, we point out the slight differences for the case where $\X$ consists of proper vertex-maximal components, making the claim applicable also to Theorem \ref{thm:properComplete2thin}.

We first construct a subgraph $\M$ $\subseteq G$ with $|E(\M)|$ equal to IE$_{\text{full}}(\X)$ as follows. For $1\leq i \leq m$, denote by $\N_i$ a {\em maximum} sized Maxwell-independent subgraph of $G_i$. Then from Lemma \ref{obs:maximal}(\ref{obs:maximal1}), we know that for any edge $e\in \H(\X)$, there is at most one $\N_i$, such that $ \{e\} \cup E(\N_i)$ is Maxwell-dependent. ({\em Note:} from Lemma \ref{obs:maximal}(\ref{obs:maximal2}), even if $\X$ is a cover by complete collection of proper vertex-maximal components, when there are no two components $G_i$ and $G_j$ s.t. $V$ $=$ $V(G_i)$ $\cup$ $V(G_j)$, it still holds that for any edge $e\in \H(\X)$, there is at most one $\N_i$, such that $ \{e\} \cup E(\N_i)$ is Maxwell-dependent.)


\medskip
\noindent Thus, edges of component $G_i$ can be divided into four parts:
\begin{itemize}
\item $\P_1^i$: the set of edges $e$ in $\H(\X)\cap E(\N_i)$ that are present in each $E(\N_j)$ for which $G_j$ contains $e$;
\item $\P_2^i$: the set of edges $e$ in $\H(\X)\cap E(\N_i)$ for which there is exactly one $\N_j$ where $e\in$ $G_j\setminus \N_j$, i.e., $ \{e\} \cup E(\N_j)$ is Maxwell-dependent;
\item $\P_3^i$: the set of edges $e$ in $\H(\X)\setminus E(\N_i)$, and present in all other $\N_j$'s, where $G_j$ contains $e$.
\item $\P_4^i$: $E(G_i)\setminus \H(\X)$.
\end{itemize}

Let $\P_k= \bigcup\limits_{i} \P_k^i$. Now we construct $\M$ as follows. First, let $V(\M):= V(G)$. Then we construct the edge set $E(\M)$ by removing all edges in $\P_2$ and $\P_3$ from $\bigcup\limits_{i=1}^m E(\N_i)$. Thus $\N_i$ $=$ $\M|_i$ $\cup$ $\P_2^i$, where $\M|_i$ denotes $\M$ restricted to $G_i$.

Now note that $|E(\M)| = \sum\limits_{i=1}^m (3|V(G_i)| - 6)
-\sum\limits_{{\{u, v\}}\in \P_1} (n_{\{u, v\}} -1)$ $-$ $\sum\limits_{{\{u, v\}}\in \P_2\cup \P_3}$ $(n_{\{u, v\}}$ $-1)$, which is exactly IE$_{\text{full}}(\X)$, since $\X$ is strong $2$-thin.
In the following we show that this number is at least rank$(G)$ by showing that $\M$ is a maximal Maxwell-independent subgraph of $G$ and using Theorem \ref{thm:main}.
\begin{enumerate}[(I)]

\item $\M$ is Maxwell-independent. Suppose not, then we can find a minimal subgraph $\M^\prime \subseteq \M$ that is Maxwell-dependent. Since $\M$ is picked in such a way that every $\M|_i$ is Maxwell-independent, we know $\M^\prime$ cannot be inside any $G_i$. Because $\M^\prime$ is minimal, we know there exists $\M^{\prime\prime} \subset \M^\prime$ that (1) contains all vertices of $\M^\prime$ and (2) is Maxwell-independent with Maxwell count $6$.
Then $\M^{\prime\prime}$ is a component that is not contained in any $G_i$, since $\M^\prime$ is not inside any $G_i$, and removing an edge from $\M^\prime$ does not make it inside any $G_i$ either. That is a contradiction to the fact that $G_1, \ldots, G_m$ is the {\em complete} collection of vertex-maximal components of $G$. ({\em Note:} this contradiction would hold even if $\X$ is a cover by complete collection of proper vertex-maximal components.)

\item $\M$ is a maximal Maxwell-independent subgraph of $G$. 
In order to show this, we first notice that for every $e\in \P_2^i$, every maximal Maxwell-independent subgraph $\N_i^\prime$ of $G_i$ contains $e$, which follows from the statements that (1) there exists a $G_j$ s.t $ \{e\} \cup E(\N_j)$ is Maxwell-dependent and (2) Lemma \ref{obs:maximal}(\ref{obs:maximal1}). ({\em Note:} from Lemma \ref{obs:maximal}(\ref{obs:maximal2}), even if $\X$ is a cover by complete collection of proper vertex-maximal components, when there are no two components $G_1$ and $G_2$ s.t. $V$ $=$ $V(G_1)$ $\cup$ $V(G_2)$, it still holds that for every $e\in \P_2^i$, every maximal Maxwell-independent subgraph of $G_i$ contains $e$.)


Suppose there is an edge $e \in E\setminus E(\M)$ such that $E(\M)\cup \{e\}$ is  Maxwell-independent. Then $(E(\M)\cup \{e\})|_i$ (which denotes $E(\M)\cup \{e\}$ restricted to $G_i$) is also Maxwell-independent. 
Since $e\in G_i$ for some $i$, we know $e\in \P_2^i, \P_3^i$ or $\P_4^i$. In fact every edge $\P_2^j$ for some $j$ is also in $\P_3^i$ for some $i$, without loss of generality, we choose a component $i$ such that $e\in \P_3^i$ or $\P_4^i$. 
Notice that there is an extension of $(E(\M)\cup \{e\})|_i$ into a maximal Maxwell-independent subgraph $\M_i^\prime$ of $G_i$, which must contain all edges in $\P_2^i$ as shown in the previous paragraph, i.e., $E(\M_i^\prime)$ contains $(E(\M)\cup \{e\})|_i$ $\cup$ $\P_2^i$. Since $\M|_i$ $\cup$ $\P_2^i$ $=$ $\N_i$, we know $E(\M_i^\prime)$ has size larger than $E(\N_i)$, which is a contradiction to the fact that $\N_i$ is a maximum sized Maxwell-independent subgraph of $G_i$. Hence $\M$ is maximal Maxwell-independent.

\end{enumerate}

\medskip\noindent
Thus we know $\M$ is a maximal Maxwell-independent set of $G$. From Theorem \ref{thm:main}, we know $|E(\M)| \geq \text{rank}(G)$. As noticed before, the IE$_{\text{full}}$ count of the cover $\X$ is equal to $|E(\M)|$, hence we have
$\sum\limits_{i=1}^m{(3|V(G_i)| - 6)}$ $-$ $\sum\limits_{\{u, v\} \in \H(\X)}$ $(n_{\{u, v\}} -1) \geq \text{rank}(G)$. 

\end{proof}

Returning to the proof of Theorem \ref{thm:complete2thin}, we first notice that Claim \ref{clm:2thin} completes the proof, when there are no edge components in the cover $\X$.

With edge components in the cover, notice that each edge component contributes $1$ to the left hand side but contributes at most $1$ to the right hand side. Thus the inequality still holds.

\end{proof}

The next theorem extends the bound in Theorem \ref{thm:complete2thin} to covers by proper vertex-maximal components.

\begin{thm}\label{thm:properComplete2thin}
For a graph $G=(V, E)$, if the complete collection $\X$ $=$ $\{e_1, \ldots$, $e_k$, $G_1, G_2, \ldots, G_m\}$ of proper vertex-maximal components forms a 2-thin cover, then the IE$_{\text{full}}$ count of the cover $\X$ is an upper bound on $\text{rank}(G)$, i.e., $$\sum\limits_{i=1}^m{(3|V(G_i)| - 6)}  -  \sum\limits_{{\{u, v\}} \in \H(\X)}(n_{\{u, v\}}-1) +k \geq \text{rank}(G).$$
\end{thm}

\begin{proof}
When $G$ is not Maxwell-rigid, the proof is the same as in Theorem \ref{thm:complete2thin}.

When $G$ is Maxwell-rigid, we first show the theorem for the case where there are no edge components. There are two further cases:
\begin{description}

\item[Case 1.] There exist two components $G_i$ and $G_j$ s.t. $V(G)$ $=$ $V(G_i)$ $\cup$ $V(G_j)$. In this case, all other non-trivial components in the cover can only be $K_3$ or $K_4$. For every edge $e$ in those components, we know (1) if $e\in$ $G_i\cup G_j$, then $e$ contributes to $0$ to both the left hand side and right hand side of the inequality; and (2) if $e\not\in$ $G_i\cup G_j$,  then $e$ contributes to $1$ to the left hand side, and $0$ or $1$ to the right hand side of the inequality. 

Thus if we can show that IE$_\text{full}$ count on $G_i\cup G_j$ is an upper bound on the rank of $G_i\cup G_j$, then the theorem holds. Note that IE$_\text{full}$ count on $G_i\cup G_j$ is equal to $3|V|-7$, and from the axiom $C5$ of abstract rigidity matroid (see \cite{graver:servatius:rigidityBook:1993}), we know $G_i\cup G_j$ is not rigid and thus rank$(G_i\cup G_j)$ is at most $3|V|-7$. Hence IE$_\text{full}$ count on $G_i\cup G_j$ is an upper bound on the rank of $G_i\cup G_j$.

\item[Case 2.] For any two components $G_i$ and $G_j$, we have $V(G)$ $\neq$ $V(G_i)$ $\cup$ $V(G_j)$. In this case, we know the cover is strong $2$-thin, since otherwise, there exist two components $G_1$ and $G_2$ whose intersection is a pair of vertices without an edge. From Lemma \ref{lem:join}(\ref{lem:join1}), we know $G_1\cup G_2$ is Maxwell-rigid. Since both $G_1$ and $G_2$ are proper vertex-maximal components, we know $V(G)$ $=$ $V(G_1)$ $\cup$ $V(G_2)$, a contradiction.

Next, we apply Claim \ref{clm:2thin} of Theorem \ref{thm:complete2thin} to complete the proof of Theorem \ref{thm:properComplete2thin} where there are no edge components.

\end{description}

Now we can consider the case with edge components in the cover and notice that each edge component contributes $1$ to the left hand side but contributes at most $1$ to the right hand side. Thus the inequality still holds.
\end{proof}

\noindent
{\bf Remark:}
(I) In fact, in Theorems \ref{thm:complete2thin} and \ref{thm:properComplete2thin}, when $G$ is not Maxwell-rigid or $G$ has at least $3$ non-trivial components
in the strong 2-thin cover $\X$, it turns out that we do not {\em need}
Theorem \ref{thm:main} to show that the IE$_{\text{full}}$ count
of the cover $\X$ is an upper bound on $\text{rank}(G)$. This is because we can show that $\M$ constructed in Theorem
\ref{thm:complete2thin} is in fact a {\em maximum-size} Maxwell-independent subgraph
of $G$. Otherwise we can find a maximal Maxwell-independent subgraph
$\M^\prime$ such that $|E(\M^\prime)| > |E(\M)|$. Then there must be some $i$ such
that $|E(\M^\prime)|_i| >|E(\M)|_i|$. We know $\P_2^i$ is Maxwell-independent in
every Maxwell-independent set of $C_i$ and since $\M^\prime|_i$ is
Maxwell-independent, hence $E(\M^\prime)|_i\cup \P_2^i$ is also
Maxwell-independent with size greater than $E(\M)|_i\cup \P_2^i$, which is
$E(\N_i)$. That is a contradiction to the fact that $\N_i$ is a maximum sized
Maxwell-independent subgraph of $C_i$. 
(II) We can use the maximum sized Maxwell-independent subgraph $\M$ constructed in Theorems \ref{thm:complete2thin} and \ref{thm:properComplete2thin} to test Maxwell-rigidity.

\section{Open Problems}
\label{conclusion}

\subsection{Extending Rank bound to Higher Dimensions}\label{sec:higher}

The definition of maximal Maxwell-independent set extends to all dimensions, leading to the following conjecture.
\begin{conjecture}\label{conj:kdim}
For any dimension $d$, the size of any maximal Maxwell-independent set gives an upper bound on the rank of the generic rigidity matroid of a graph $G$.
\end{conjecture}

\medskip\noindent
Moreover, the definition of 2-thin component graphs can also be extended to $d$ dimensions.

\begin{dfn}\label{dfn:dDimcomponentGraph}
Given $G=(V, E)$, let $\X=\{G_1, G_2,$ $ \ldots,$ $ G_m\}$ be a $(d-1)$-thin cover of $G$, i.e., $|V(G_i) \cap V(G_j) |$ $\leq$ $d-1$ for all $1 \leq i < j \leq m$. The {\em $(d-1)$-thin component graph} $\C_{\X}$ of $G$ contains a {\em component node} for each subgraph induced by $G_i$ in $\C_{\X}$ and whenever $G_i$ and $G_j$ share
a complete graph $K_{d-1}$ in $G$, their corresponding component nodes in $\C_{\X}$ are connected via an {\em edge
node}. The {\em degree} of a component node is defined to be the number of its adjacent edge nodes.
\end{dfn}

\medskip\noindent
To show Conjecture \ref{conj:kdim}, Proposition \ref{obs:vmmr} will have to be shown for $(d-1)$-thin covers and it is sufficient to show that the  $(d-1)$-thin component graphs of Maxwell-independent
sets are generalized partial ${d+1 \choose 2}$-trees. However, we conjecture one possible generalization of the
strongest bound that we are able to show in the proof of Lemma \ref{lem:comb}(\ref{thm:degree}).

\begin{conjecture}\label{conj:kdimTree}
For a Maxwell-independent graph with a $(d-1)$-thin cover $\X$ in $d$ dimensions
the average degree of the component nodes of any subgraph of the $(d-1)$-thin component graph is strictly
smaller than $d+1$.
\end{conjecture}

For $d=2$ this bound says that for Maxwell-independent sets, the average degree of the component nodes in the component graph is at most $2$. For $d=3$, however, we do not know of an example where all nodes have degree $\ge 3$. In fact, we do not even know of an example with average degree $\ge 3$. We state this as a conjecture for generalized body-hinge frameworks.
 
 \begin{conjecture}\label{conj:bodyhinge}
 In a $3$-dimensional independent generalized body-hinge framework (where several bodies can meet at a hinge and several hinges can share a vertex), the average number of hinges per body is less than $3$.
 \end{conjecture}

 \medskip\noindent
 Lemma \ref{lem:comb}(\ref{thm:degree}) shows that there is no subgraph of the 2-thin component graph where each
 component node has at least $4$ shared edges.
 A natural question is whether the counts for the so-called ``identified'' body-hinge
 frameworks can be used \cite{tay:rigidity1984, WhWh87,KatohTanigawa2009, Tanigawa:2012}, treating the component nodes as bodies
 and the shared edges as hinges. However, while identified body-hinge
 frameworks account for {\em several} component nodes sharing an edge (as we have here), generalized body-hinge structures may additionally have shared edges that have common vertices, hence the
 generic, identified body-hinge counts may not apply.

\subsection{Stronger Versions of Independence}\label{sec:stronger}
Even for Maxwell-independent graphs, the rank bounds of our Theorem \ref{thm:main}
can be arbitrarily bad. Even a simple
example of 2 bananas without the hinge edge has a single maximal
Maxwell-independent set of size 18 (which is the bound given by all of our
theorems), but its rank is only 17. Another example is the so-called ``$n$-banana'': it is formed by joining $n$ $K_5$'s on an edge and then removing that shared edge. In the $n$-banana, the whole graph is Maxwell-independent, so itself is the unique maximal Maxwell-independent set. This maximal Maxwell-independent set exceeds the rank of the $3$-dimensional generic rigidity matroid of $n$-banana by $n-1$.

\medskip\noindent 
Theorem \ref{thm:propermaximal} give alternative upper bounds for Maxwell-independent graphs. (In fact, Theorem \ref{thm:propermaximal} leads to a recursive method of obtaining a rank bound by recursively decomposing the graph into proper vertex-maximal
components. As one consequence, it gives an alternative, much simpler proof of
correctness for an existing algorithm called the Frontier Vertex algorithm (first version) that is based on
this decomposition idea as well as other ideas in this
chapter such as the component graph \cite{bib:survey}.)

\medskip\noindent
A natural open problem is to improve the bound in Theorem \ref{thm:main} directly by considering other
notions of independence that are stronger than Maxwell-independence. (Algorithms in \cite{sitharam:zhou:tractableADG:2004, bib:survey}
suggest and use stronger notions than Maxwell-independence,
but the algorithms usually
use some version of an inclusion-exclusion formula. They do not
provide explicit maximal sets of edges satisfying the stronger notions of
Maxwell-independence.
Neither do they prove that all such sets provide good bounds.)

\subsection{Bounds for Maxwell-Dependent Graphs Using 2-Thin Covers}\label{sec:Maxwelldependent}

\medskip\noindent 
While Theorem \ref{thm:propermaximal} gives a strong rank bound for Maxwell-independent graphs, Theorem \ref{thm:complete2thin} and Theorem \ref{thm:properComplete2thin} give much weaker bounds for Maxwell-dependent graphs because a collection of (proper) vertex-maximal, Maxwell-rigid subgraphs may be far from being a 2-thin cover. For example, in Figure \ref{fig:n2thin} we have $3$ $K_5$'s and the neighboring $K_5$'s share an edge with each other. There are two vertex-maximal, Maxwell-rigid subgraphs, each of which consists of $2$ $K_5$'s with a shared edge.

\begin{center}
\begin{figure}[!h]
\begin{center}
\scalebox{0.8}[0.8]{\includegraphics{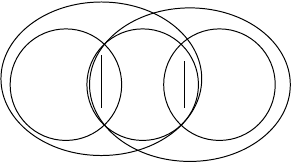}}
\end{center}
\caption{A cover of vertex-maximal components that is not 2-thin. The circles are $K_5$'s and the two larger ellipses are vertex-maximal, Maxwell-rigid subgraphs that form the cover.}
\label{fig:n2thin}
\end{figure}
\end{center}

\medskip\noindent
While many other 2-thin covers exist, the completeness as well as (proper) vertex-maximality
are important ingredients in the proofs of these theorems.
One possibility is to use 2-thin covers that are a subcollection of
(proper) vertex-maximal, Maxwell-rigid subgraphs. Another is to use collections of not necessarily vertex-maximal, but Maxwell-rigid subgraphs in which no proper subcollection of $2$ or more subgraphs has a Maxwell-rigid union.

\medskip\noindent
Another notion that can be used involves the following definition of {\em strong Maxwell-rigidity}:
\begin{dfn}
A graph $G=(V, E)$ is \emph{strong Maxwell-rigid} if for all maximal Maxwell-independent edge sets $E^\prime \subseteq E$, we have $|E^\prime|
= 3|V(E^\prime)|-6$.
\end{dfn}

It is tempting to use the approach in Theorem \ref{thm:complete2thin} to show that the IE$_{\text{full}}$ count for a cover by vertex-maximal, strong Maxwell-rigid subgraphs is a new upper bound on the rank. We conjecture the $2$-thinness of the cover, which is a crucial property explored in proving Theorem \ref{thm:complete2thin}.
\begin{conjecture}\label{conj:strong2co}
Any cover of a graph by a collection of vertex-maximal, strong Maxwell-rigid subgraphs is a 2-thin cover.
\end{conjecture}

\medskip\noindent
However, the idea in the proof of Theorem \ref{thm:complete2thin} will not work because the set $\M$, constructed in the proof of Theorem \ref{thm:complete2thin} that is of size equal to the IE$_{\text{full}}$ count, 
can now be of smaller size than {\em any} maximal Maxwell-independent set of $G$ as in the example of Figure \ref{smr_counter}.

\medskip\noindent
{\bf Example(Figure \ref{smr_counter}):} there are
five rings of $K_5$'s, where each ring consists of $7$ $K_5$'s. In the graph,
every $K_5$ is a vertex-maximal strong Maxwell-rigid subgraph, and the
IE$_{\text{full}}$
count for the cover $\X$ is $(3*5-6 )* (6*5+1) - 5*5 -10 = 244$. Here the $(6*5+1)$ is the number of $K_5$'s and $5*5 + 10$ is the total number of shared edges. But if we take $9$
edges in every $K_5$ except $T$ such that the missing edges are not shared, then we obtain a set $\M^\prime$ that is Maxwell-dependent. From $\M^\prime$ we drop one edge $e$ of $T$ and add one missing edge $f$ to the $K_5$ that shares $e$ with $T$. Then we get a set $\M^{\prime\prime}$ that is a {\it minimum-size} maximal Maxwell-independent set of $G$. The size of $\M^{\prime\prime}$ is $(6*9-5)*5=245$, where $6*9-5$ is the number of edges in each ring, not counting the edges in $T$ that are unshared in that ring. 

\medskip\noindent 
Hence in the Figure \ref{smr_counter} example, the IE$_{\text{full}}$ count is less than the size of any maximal Maxwell-independent set, so the latter cannot be used as a bridging inequality
as in Theorem \ref{thm:complete2thin}.
However, the IE$_{\text{full}}$ count does seem to give a direct upper bound on the rank (it is equal to the rank) hence a different proof idea might yield the required bound on rank.
\begin{center}
\begin{figure}[!h]
\begin{center}
\includegraphics[width=.5\textwidth]{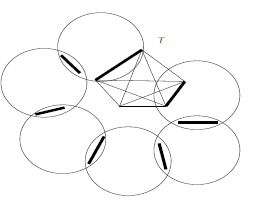}
\end{center}
\caption[A counterexample to show that IE$_{\text{full}}$ count of cover $\X$ by vertex-maximal, {\em strong} Maxwell-rigid subgraphs turns out to be smaller than the size of any maximal Maxwell-independent set.]{A counterexample to show that IE$_{\text{full}}$ count of cover $\X$ by vertex-maximal, {\em strong} Maxwell-rigid subgraphs turns out to be smaller than the size of any maximal Maxwell-independent set.
Start with a $K_5$, denoted $T$. Each of $5$ pairs of edges of $T$ is extended into a ring of $7$ $K_5$'s, where each ring is formed by closing a chain of $K_5$'s where the neighboring $K_5$'s share an edge (bold) with each other. In each
of the $5$ rings, every $K_5$ shares an edge with each of its two neighboring $K_5$'s and these two edges are non-adjacent. Note that in the figure, only one of the five rings is shown.}
\label{smr_counter}
\end{figure}
\end{center}

\subsection{Algorithms for Various Maximal Maxwell-Independent Sets}
So far the emphasis has been to find good upper bounds on rank and Theorem \ref{thm:main} shows that the {\it minimum-size} maximal Maxwell-independent set of a graph $G$ is at least $\text{rank}(G)$. A natural open problem is to give an algorithm that constructs a minimum-size,
maximal Maxwell-independent set of an arbitrary graph.

\medskip\noindent
Note that Maxwell-rigidity requires the {\it maximum} Maxwell-independent set to be of size $\geq 3|V|-6$. Although the maximum Maxwell-independent set is trivially as big as the rank (and is not directly relevant to finding good bounds on rank), covers by Maxwell-rigid components have played a role in some of the theorems above (Theorems \ref{thm:propermaximal}, \ref{thm:complete2thin}, \ref{thm:properComplete2thin}) that give useful bounds on rank. Recall that Hendrickson \cite{Hendrickson92conditionsfor} gives an algorithm to test $2$-dimensional Maxwell-rigidity by finding a maximal Maxwell-independent set that is automatically maximum for $d=2$. While an extension of Hendrickson \cite{Hendrickson92conditionsfor} to $3$ dimensions given in \cite{andrewThesis} finds {\em some} maximal Maxwell-independent set, it is not guaranteed to be maximum (or minimum). Thus another question of interest is whether maximum Maxwell-independent sets can be characterized in some natural way.

\subsection*{Acknowledgement}
We thank Bill Jackson and an anonymous reviewer for a careful reading and many constructive suggestions to improve the presentation.

\bibliographystyle{plain}

\bibliography{biblio}
\end{document}